\documentclass[12pt]{amsart}   
\usepackage{amssymb,hyperref}

\theoremstyle{plain}
\newtheorem{theorem}{Theorem}[section]

\newtheorem{proposition}[theorem]{Proposition}
\newtheorem{corollary}[theorem]{Corollary}
\theoremstyle{definition}
\newtheorem{definition}[theorem]{Definition}

\newtheorem{remark}[theorem]{Remark}

\DeclareMathOperator{\BSub}{BSub}
\DeclareMathOperator{\Aut}{Aut}

\newcommand{\mc}{\mathcal}

\newcommand{\ul}{\underline}

\title{Topos quantum theory with short posets}

\author{John Harding}
\address{New Mexico State University}
\email{hardingj@nmsu.edu}

\author{Chris Heunen}
\address{University of Edinburgh}
\email{chris.heunen@ed.ac.uk}

\begin{document}

\begin{abstract}
  Topos quantum mechanics, developed by Isham {\em et.\ al.}~\cite{Measures,Doring,Main,Flori1,Flori2,Isham,Isham2,Isham3,Isham4}, creates a topos of presheaves over the poset $\mc{V}(\mc{N})$ of Abelian von Neumann subalgebras of the von Neumann algebra $\mc{N}$ of bounded operators associated to a physical system, and established several results, including: (a) a connection between the Kochen-Specker theorem and the non-existence of a global section of the spectral presheaf; (b) a version of the spectral theorem for self-adjoint operators; (c) a connection between states of $\mc{N}$ and measures on the spectral presheaf; and (d) a model of dynamics in terms of $\mc{V}(\mc{N})$. We consider a modification to this approach using not the whole of the poset $\mc{V}(\mc{N})$, but only its elements $\mc{V}(\mc{N})^*$ of height at most two. This produces a different topos with different internal logic. However, the core results (a)--(d) established using the full poset $\mc{V}(\mc{N})$ are also established for the topos over the smaller poset, and some aspects simplify considerably. Additionally, this smaller poset has appealing aspects reminiscent of projective geometry. 
  \\[\baselineskip]
\end{abstract}

\maketitle

\section{Introduction}

Isham and Butterfield \cite{Isham,Isham2,Isham3,Isham4} introduced a topos approach to quantum mechanics and showed that the Kochen-Specker theorem is equivalent to the non-existence of a global section of a certain presheaf. This program was further developed by Butterfield, D\"oring, de Groote, Flori, Hamilton, and Isham \cite{Measures,Doring3,Doring4,Doring2,Main,Main2,Flori1,Flori2}; and a related topos approach was introduced by Heunen, Landsman, and Spitters \cite{HLS,HLS2,HLS3,Heunen}. In \cite{Flori,Wolters} the approaches are compared, and the books \cite{Flori1,Flori2} are devoted to the subject. 

In this approach, a von Neumann algebra $\mathcal{N}$ is associated to a quantum system as in standard quantum mechanics. Then $\mathcal{V}(\mathcal{N})$ is the poset of unital abelian von Neumann subalgebras of $\mathcal{N}$. The fundamental object is the topos of presheaves of sets over $\mathcal{V}(\mathcal{N})$. In particular, this is applied in the case when $\mathcal{N}$ is the von Neumann algebra $\mathcal{B}(\mathcal{H})$ of bounded operators on a Hilbert space $\mathcal{H}$, and in this case $\mathcal{V}(\mathcal{N})$ is denoted $\mathcal{V}(\mathcal{H})$. 

The spirit of this topos approach is that elements of $\mathcal{V}(\mathcal{N})$ give classical ``snapshots'' of the quantum system. These classical snapshots are then glued together to form the various presheaves used in the topos approach. It is the purpose of this note to show that the primary results of the topos approach \cite{Main} are retained when one considers the topos of presheaves over the poset $\mathcal{V}(\mathcal{N})^*$ of elements of $\mathcal{V}(\mathcal{N})$ of height at most two. In rough terms, rather than use all classical snapshots, classical ``glimpses'' of the system suffice. 

There are several implications of this line of study. Clearly the poset $\mathcal{V}(\mathcal{N})^*$ is far simpler than $\mathcal{V}(\mathcal{N})$, and this simplification is reflected in aspects of the development of the quantum approach over $\mc{V}(\mc{N})^*$. For instance, the spectral presheaf associates spectra to each $V\in\mc{V}(\mc{N})$. In general, these are compact Hausdorff spaces. However, the spectral presheaf over $\mc{V}(\mc{N})^*$ associates spectra only to those $V\in\mc{V}(\mc{N})^*$, and these are discrete spaces with at most 3 elements. So clopen subobjects of the spectral presheaf over $\mc{V}(\mc{N})^*$ are simply the subobjects. In effect, the topology plays no role. The spirit of $\mc{V}(\mc{N})$ is to take classical snapshots of the system, the spirit of $\mc{V}(\mc{N})^*$ is to take small finite-dimensional classical snapshots. 

Considering the topos approach in the setting of $\mc{V}(\mc{N})^*$ also provides an alternate viewpoint from which one can assess the topos approach more broadly. The same core theorems are established in the topos over $\mc{V}(\mc{N})$ and that over $\mc{V}(\mc{N})^*$, and they are established by similar methods. Yet the two toposes are not the same, and their logics have different equational properties with the logic of the topos over $\mc{V}(\mc{N})^*$ being stronger than that of $\mc{V}(\mc{N})$. This leads one to ask what is the role of this logic in terms of the quantum system being represented. 

In connection with the above point, there are differences between the two approaches. In \cite{Measures} a bijective correspondence was given between the states of $\mc{N}$ and the finitely additive measures on $\mc{V}(\mc{N})$. Such measures were described as certain presheaves on the clopen subobjects of the spectral presheaf. This carries over to the $\mc{V}(\mc{N})^*$ setting. However, a further characterization of the normal states was given in terms of those measures that satisfied a certain additional ``local'' property, and this does not carry to the $\mc{V}(\mc{N})^*$ setting. It is not now known whether there are any results in the $\mc{V}(\mc{N})$ setting expressible in ``global'' terms that do not carry to the $\mc{V}(\mc{N})^*$ setting. 

The topos approach in the $\mc{V}(\mc{N})^*$ setting has some additional advantages. The poset $\mc{V}(\mc{N})^*$ is not only much smaller than $\mc{V}(\mc{N})$, but it is tractable to work with. In \cite{HHLN,BertJohn} it was shown that $\mc{V}(\mc{N})^*$ can be treated graphically and has similarities reminiscent of projective geometry. This extends to a treatment of morphisms between the projections of von Neumann algebras $\mc{N}$ and $\mc{M}$ that become certain order preserving maps between the posets $\mc{V}(\mc{N})^*$ and $\mc{V}(\mc{M})^*$, and therefore certain geometric morphisms between their toposes. This may be of interest in the larger program of studying multiple systems indicated in \cite{Main}. 

This note is arranged in the following way. Section~\ref{sec:review} briefly reviews the major features of the topos approach. We use the same notation and terminology as in~\cite{Main}, and this section is provided to make the paper more easily read. Section~\ref{sec:shorttopos} shows how the major results of~\cite{Main} carry over to the setting of $\mathcal{V}(\mathcal{N})^*$. This amounts largely to isolating key features from the more complex setting of $\mathcal{V}(\mathcal{N})$. 
Section~\ref{sec:automorphisms} establishes that automorphisms of $\mathcal{V}(\mathcal{N})^*$ behave exactly the same as automorphisms of $\mathcal{V}(\mathcal{N})$. 
Finally, Section~\ref{sec:conclusion} contains some concluding remarks, including a discussion of the role of von Neumann algebras as opposed to the more general C*-algebras~\cite{HLS}. 

\section{A review of the topos approach}\label{sec:review}

A detailed account of the topos approach of Isham, Butterfield and D\"oring is given in \cite{Main}. We follow that with the same notation in this note. We comment that \cite{Main} is largely similar to \cite{Main2}, but we follow \cite{Main} as it may be more easily obtained. We briefly review the outline here, but for details the reader should consult \cite{Main}. 

Let $\mc{N}$ be a von Neumann algebra and $\mc{V}(\mc{N})$ be the poset of its abelian von Neumann subalgebras. The paper \cite{Main} works primarily in the setting of the von Neumann algebra 
$\mc{B}(\mc{H})$ of bounded operators of a Hilbert space $\mc{H}$, and denotes $\mc{V}(\mc{B}(\mc{H}))$ by $\mc{V}(\mc{H})$. The topos of presheaves over $\mc{V}(\mc{N})$ is the fundamental setting. The key in defining the pertinent presheaves is the notion of {\em daseinisation} \cite[p.56]{Main}. For each $V\in\mc{V}(\mc{N})$ let $\mc{P}(V)$ be the complete Boolean algebra of projections of $V$. Daseinisation $\delta$ associates to any projection $\hat{P}$ of $\mc{N}$ and any $V\in\mc{V}(\mc{N})$ the smallest projector in $\mc{P}(V)$ lying above $\hat{P}$. 
\begin{equation} \label{daseinisation}
\delta(\hat{P})_V=\bigwedge\{\hat{\alpha}\in\mc{P}(V):\hat{P}\leq\hat{\alpha}\}
\end{equation}

The {\em outer presheaf} or {\em coarse graining presheaf} \cite[p.57]{Main} $\ul{O}$ is defined on objects $V\in\mc{V}(\mc{N})$ and morphisms $i_{V'V}$ for $V'\subseteq V$ as follows:
\begin{equation} \label{outer}
\ul{O}_V = \mc{P}(V)\quad \mbox{ and }\quad
\ul{O}(i_{V'V})\mbox{ takes $\hat{P}\in\mc{P}(V)$ to $\delta(\hat{P})_{V'}\in\mc{P}(V')$}
\end{equation}

The spectrum of an abelian von Neumann algebra $V$ is the collection of multiplicative linear functionals $\lambda:V\to\mathbb{C}$. It forms a compact Hausdorff space. The {\em spectral presheaf} \cite[p.62]{Main} $\ul{\Sigma}$ is as follows: 
\begin{equation} \label{spectral}
\ul{\Sigma}_V = \mbox{ spectrum of }V\quad \mbox{ and } \quad
\ul{\Sigma}(i_{V'V}) \mbox{ takes $\lambda$ to its restriction $\lambda|_{V'}$}
\end{equation}
Any projection $\hat{P}\in\mc{N}$ gives a global element $\delta(\hat{P})$ of $\ul{O}$ \cite[p.57]{Main}. In contrast, in the case of $\mc{B}(\mc{H})$ where dim $\mc{H} > 2$ we have the following. 

\begin{theorem}(\cite[p.62]{Main})\label{Kochen-Specker}
The Kochen-Specker theorem for a Hilbert space $\mathcal{H}$ is equivalent to the spectral presheaf $\ul{\Sigma}$ of $\mc{V}(\mc{H})$ having no global sections. 
\end{theorem}

The ``values object'' for the topos approach is not completely settled, and a number of related presheaves play a role. To begin, there is the constant presheaf $\ul{\mathbb{R}}$ whose constant value is the ordinary real numbers. This is the internal real numbers object. The presheaves that play a larger role are variants of $\ul{\mathbb{R}^\preceq}$  \cite[Defns.~8.1--8.3]{Main}. For $V\in\mathcal{V}(\mathcal{N})$ let ${\downarrow}V$ be the principle downset of $V$ and $\mc{OP}({\downarrow}V,\mathbb{R})$ be the set of order preserving functions $\alpha:{\downarrow}V\to\mathbb{R}$. Then 
\begin{equation} \label{OP}
\ul{\mathbb{R}^\preceq}_V = \mc{OP}({\downarrow}V,\mathbb{R})\quad \mbox{ and } \quad
\ul{\mathbb{R}^\preceq}(i_{V'V}) \mbox{ takes $\alpha$ to its restriction $\alpha |_{{\downarrow}V'}$}
\end{equation}
Then $\ul{\mathbb{R}^\succeq}$ is defined similarly through order inverting functions \cite[Defn.~8.2]{Main} and $\ul{\mathbb{R}^\leftrightarrow}$ \cite[Defn.~8.3]{Main} where $\ul{\mathbb{R}^\leftrightarrow}_V$ is all ordered pairs $(\alpha,\beta)$ of an order preserving function $\alpha$ and order inverting function $\beta$ from ${\downarrow} V$ to $\mathbb{R}$ with $\alpha\leq \beta$ pointwise. This is also known as the interval domain~\cite[Sec.~1.6]{HLS}.

\begin{theorem}\label{SpecThm}(\cite[Thm.~8.2]{Main}
For a self adjoint element $\hat{A}$ of $\mathcal{V}$, there is a natural transformation $\check{\delta}(\hat{A})$ from the spectral presheaf $\ul{\Sigma}$ to the value presheaf $\ul{\mathbb{R}^\leftrightarrow}$. 
\end{theorem}

This provides an analog of the spectral theorem similar to that in the classical case --- to each observable $\hat{A}$ of the system, there is a function from the states to the values object. This $\check{\delta}$ is defined through an extension of daseinisation to self-adjoint operators. Under the ``spectral order'' $\preceq$ \cite[p.~95]{Main} the self-adjoint operators $\mathcal{N}_{sa}$ of $\mathcal{N}$ form a complete lattice. For a self adjoint $\hat{A}\in\mathcal{N}_{sa}$ the outer daseinisation $\delta^o(\hat{A})$ is defined for $V\in\mathcal{V}(\mathcal{N})$ by
\begin{equation}\label{outerdaseinisation}
\delta^o(\hat{A})_V = \bigwedge\{\hat{B}\in V_{sa}:\hat{A}\preceq\hat{B}\}
\end{equation}
The inner daseinisation $\delta^i(\hat{A})$ is defined similarly via join as the largest in $V_{sa}$ beneath $\hat{A}$. Since each $V\in\mathcal{V}(\mathcal{N})$ is abelian, the self-adjoint operator $\delta^o(\hat{A})_V$ of $V$ provides a map $\overline{\delta^o(\hat{A})_V}:\ul{\Sigma}_V\to\mathbb{R}$ from the spectrum $\ul{\Sigma}_V$ of $V$ to the reals. Then for each $\lambda\in\ul{\Sigma}_V$ there is a map \cite[p.103]{Main} $\check{\delta}^o(\hat{A})_V(\lambda):{\downarrow}V\to\mathbb{R}$ whose value at $V'\subseteq V$ is $\overline{\check{\delta^o}(\hat{A})_{V'}}(\lambda|_{V'})$. With similar considerations for inner daseinisation, the components for the natural transformation $\check{\delta}(\hat{A})$ of Theorem~\ref{SpecThm} are given by \cite[p.~106]{Main}
\begin{equation} \label{SAdaseinisation}
\check{\delta}(\hat{A})_V = (\check{\delta}^i(\hat{A})_V(\,\cdot\,),\check{\delta}^o(\hat{A})_V(\,\cdot\,)):\ul{\Sigma}_V\to\ul{\mathbb{R}^\leftrightarrow}_V
\end{equation}

This is one of the primary results of the topos approach, showing that the observables, state space, and values object behave in a classical way. We discuss one further result from the topos approach. For this, a further notion is required. 

A subobject $\ul{S}$ of the spectral presheaf is {\em clopen} if $\ul{S}_V$ clopen in $\ul{\Sigma}_V$ for each $V$. The standard relation between projections of a commutative von Neumann algebra and clopen subsets of its spectrum provides \cite[p.63]{Main} that each projection $\hat{P}$ of $\mc{N}$ gives a clopen subobject $\ul{\delta(\hat{P})}$ whose value at $V$ is the clopen subset of $\ul{\Sigma}_V$ associated to the projection $\delta(\hat{P})_V$ of $V$. It is shown \cite[p.62]{Main} that the clopen subobjects Sub$_{cl}(\ul{\Sigma})$ of $\ul{\Sigma}$ form a Heyting algebra, and that there is a presheaf $P_{cl}\ul{\Sigma}$ whose global elements are these clopen subobjects. 

\begin{definition} (\cite[p.~7]{Measures})\label{defn:measure}
A measure or valuation $\mu$ on the presheaf $P_{cl}\ul{\Sigma}$ is a mapping $\mu:P_{cl}\ul{\Sigma}\to\Gamma\ul{[0,1]^\succeq}$ with $\mu(\ul{\Sigma})=1$ and $\mu(\ul{S_1}\vee\ul{S_2})+\mu(\ul{S_1}\wedge\ul{S_2})=\mu(\ul{S_1})+\mu(\ul{S_2})$.
\end{definition}

In this definition, the meet and join operations are those of the Heyting algebra $P_{cl}\ul{\Sigma}$ and the addition is the addition of $\Gamma\ul{[0,1]^\succeq}$ given by pointwise sum of increasing and decreasing functions. In \cite{Measures} the following is established. 

\begin{theorem} (\cite[Theorem~IV.1]{Measures}) \label{measures}
For any von Neumann algebra $\mc{N}$ with no direct summand of type $I_2$ there is a bijection between the set of states of $\mc{N}$ and the set of measures on the clopen subobjects of the spectral presheaf $\ul{\Sigma}$ of $\mc{V}(\mc{N})$. 
\end{theorem}

In \cite[Cor.~IV.2]{Measures}, it is also shown that the normal states of $\mc{N}$ correspond to the locally $\sigma$-additive measures \cite[Eqn.~(42)]{Measures} on $P_{cl}\ul{\Sigma}$, i.e.  those that are $\sigma$-additive in each component $V$. 

\section{The topos approach over $\mathcal{V}(\mathcal{N})^*$}\label{sec:shorttopos}

In this section, we discuss the results of the previous section in the setting of the topos over $\mathcal{V}(\mathcal{N})^*$, the elements of height at most two in $\mathcal{V}(\mathcal{N})$. Here, all presheaves of the previous section retain their meaning, but we restrict their domain to the subposet $\mathcal{V}(\mathcal{N})^*$ of $\mathcal{V}(\mathcal{N})$. Rather than work directly with the poset $\mathcal{V}(\mathcal{N})^*$, we work with an isomorphic poset constructed through the Boolean subalgebras of the projection lattice of $\mathcal{N}$. 

\begin{definition}
For a von Neumann algebra $\mathcal{N}$, let $\mathcal{P}(\mathcal{N})$ be its lattice of projections, $\BSub(\mathcal{N})$ be the poset of complete Boolean subalgebras of $\mathcal{P}(\mathcal{N})$, and $\BSub(\mathcal{N})^*$ be the elements of height at most two in $\BSub(\mathcal{N})$.
\end{definition}

It is well known that if $V\in\mathcal{V}(\mathcal{N})$, then $\mathcal{P}(V)$ is a complete Boolean subalgebra of $\mathcal{P}(\mathcal{N})$. Conversely, if $B$ is a complete Boolean subalgebra of $\mathcal{P}(\mathcal{N})$ then its double commutant $\mathcal{P}(V)''$ in $\mathcal{N}$ is an abelian von Neumann subalgebra of $\mathcal{N}$. Further, $\mathcal{P}(V)''=V$ and $\mathcal{P}(B'')=B$. This establishes the following. 

\begin{proposition}
For a von Neumann algebra $\mathcal{N}$, the posets $\mathcal{V}(\mathcal{N})$ and $\BSub(\mathcal{N})$ are isomorphic, and the posets $\mathcal{V}(\mathcal{N})^*$ and $\BSub(\mathcal{N})^*$ are isomorphic. Further, the elements of $\BSub(\mathcal{N})^*$ are exactly the Boolean subalgebras of $\mc{P}(\mathcal{N})$ that have at most 8 elements. 
\end{proposition}

Each of the presheaves of the previous section has a restriction to $\mathcal{V}(\mathcal{N})^*$, and these can be realized in an equivalent way as presheaves over $\BSub(\mathcal{N})^*$. We discuss how the restriction of the spectral presheaf can be realized as a presheaf $\ul{\Sigma}^*$ over $\BSub(\mathcal{N})^*$.

For $V\in\mathcal{V}(\mathcal{N})$, its spectrum is the set of all multiplicative linear functionals $\lambda:V\to\mathbb{C}$. Each such $\lambda$ is determined by its restriction $\lambda|_{\mathcal{P}(V)}$ to the projections of $V$, this restriction is a Boolean algebra homomorphism $\lambda|_{\mathcal{P}(V)}:\mathcal{P}(V)\to 2$, and each Boolean algebra homomorphism from $\mathcal{P}(V)$ to $2$ arises this way. Thus the spectrum of $V$ is realized as the Stone space of $\mathcal{P}(V)$. Define the presheaf $\ul{\Sigma}^*$ on $\BSub(\mathcal{N})^*$ by
\begin{equation}\label{ShortSpectral}
\ul{\Sigma}^*_B = \mbox{ Stone spectrum of }B\quad \mbox{ and } \quad
\ul{\Sigma}^*(i_{B'B}) \mbox{ takes $\lambda$ to its restriction $\lambda|_{B'}$}
\end{equation}
Note that this spectral presheaf $\ul{\Sigma}^*$ is considerably simpler than $\ul{\Sigma}$ both in the fact that its domain is a much simpler poset, and that the objects $\ul{\Sigma}^*$ are Stone spaces of at most 8-element Boolean algebras, hence at most 3-element sets with the discrete topology, rather than infinite compact Hausdorff spaces. We use this spectral sheaf in 
the following version of Theorem~\ref{Kochen-Specker}.

\begin{theorem}\label{ShortKochen-Specker}
The Kochen-Specker theorem for a Hilbert space $\mathcal{H}$ of dim $>2$ is equivalent to the spectral presheaf $\ul{\Sigma}^*$ of $\mc{V}(\mc{H})^*$ having no global sections. 
\end{theorem}

\begin{proof}
Rephrasing, the statement means that from a global section $s$ of $\ul{\Sigma}^*$ one can construct a finitely additive 0,1-valued measure $\sigma$ on the projections of $\mathcal{H}$, and from a finitely additive 0,1-valued measure $\sigma$ on the projections of $\mathcal{H}$ one can construct a global section $s$ of $\ul{\Sigma}^*$. We work with the poset $\BSub(\mathcal{B}(\mathcal{H}))^*$, which we denote $\BSub(\mathcal{H})^*$, and the realization of the spectral presheaf given in (\ref{ShortSpectral}). 

Suppose $s$ is a global section of $\ul{\Sigma}^*$. Then for each $B\in\BSub(\mathcal{H})^*$ we have an element $s_B$ in the Stone space of $B$, and if $B'\subseteq B$ then $s_{B'}$ is the restriction $s_B|B'$ of the homomorphism $s_B:B\to 2$. For each projection $p$ of $\mathcal{H}$ different from 0,1 we have that $B_p=\{0,p,p',1\}$ is a 4-element Boolean subalgebra of $\mathcal{P}(\mathcal{H})$. Define $\sigma:\mathcal{P}(\mathcal{H})\to 2$ by setting $\sigma(p)=s_{B_p}(p)$ for $p\neq 0,1$ and setting $\sigma(0)=0$, and $\sigma(1)=1$. 

To see that $\sigma$ is a finitely additive measure, we need only show that if $p,q$ are orthogonal projections, then $\sigma(p\vee q)=\sigma(p)+\sigma(q)$. It is enough to show this when neither $p,q$ is $0,1$. But then there is an 8-element Boolean subalgebra $B$ of $\mathcal{P}(\mathcal{H})$ containing both $p,q$. Since $s_{B_p}$ is the restriction of $s_B$ we have that $\sigma(p)=s_{B_p}(p)=s_B(p)$. Similarly, $\sigma(q)=s_B(q)$ and $\sigma(p\vee q)=s_B(p\vee q)$. So our result follows since $s_B$ is a Boolean algebra homomorphism. 

Conversely, suppose $\sigma:\mathcal{P}(\mathcal{H})\to 2$ is a finitely additive measure. Then for each Boolean subalgebra $B$ of $\mathcal{P}(\mathcal{H})$ we have $\sigma|B$ is a Boolean algebra homomorphism. So we obtain a global section $s$ of $\ul{\Sigma}^*$ by setting $s_B = \sigma|B$. 
\end{proof}

\begin{remark}
The proof of Theorem~\ref{ShortKochen-Specker} essentially just extracts the pertinent details from the proof of Theorem~\ref{Kochen-Specker}. That one can do this is the essential point. The proof ``lives'' at the level of $\mathcal{V}(\mathcal{N})^*$, and the vast portion of $\mathcal{V}(\mathcal{N})$ plays no role. 
\end{remark}

\begin{remark}
Suppose $P$ is a subposet of $\mathcal{V}(\mathcal{N})$. One can always restrict presheaves over $\mathcal{V}(\mathcal{N})$ to presheaves over $P$. The dangers in relying on restrictions to move results such as those of the previous section from one topos of presheaves to another is that restriction may not be a one-one or onto correspondence between presheaves of $\mathcal{V}(\mathcal{N})$ and presheaves of $P$. 

If we let $P$ be all elements of height one in $\mathcal{V}(\mathcal{N})$, then the analog of Theorem~\ref{ShortKochen-Specker} fails for presheaves over $P$. In this case $P$ consists of an antichain with one element for each Boolean subalgebra $B_p=\{0,p,p',1\}$ of $\mathcal{P}(\mathcal{H})$. The spectral presheaf is an indexed family of 2-element sets. Each choice function of the product of these two element sets gives a presheaf of 1-element sets that yields a global section of the spectral presheaf. Including elements of height 2 in $P$ provides an additional restriction that eliminates these unwanted presheaves. 

Restriction from presheaves of $\mathcal{V}(\mathcal{N})$ to presheaves of $\mathcal{V}(\mathcal{N})^*$ is not one-one, and we do not know if it is onto. However, restriction in this setting is sufficiently well behaved to preserve the results of the previous setting, and this restricted setting provides much simpler presheaves. 
\end{remark}

\begin{remark}
A question related to the existence of a 0,1-valued measure on the projection lattice $\mc{P}(\mc{H})$ was raised in \cite{JohnDerek}. Does there exist a non-constant $\mathbb{Z}_2$-valued state on $\mc{P}(\mc{H})$ when $\mc{H}$ has dimension three? This amounts to assigning to the atoms $P$ of $\mc{P}(\mc{H})$ values 0 or 1 such that each pairwise orthogonal triple of atoms has an even number assigned 1. 

We remark that this can be formulated as the existence of a global section of a presheaf $\ul{G}$ of groups on $\BSub(\mc{H})$ (which in dimension 3 is equal to $\BSub(\mc{H})^*$). For each 4-element Boolean subalgebra $B$ let $\ul{G}_B$ be $\mathbb{Z}_2$ and for each 8-element Boolean subalgebra let $\ul{G}_B$ be the Klein 4 group $\mathbb{K}$. There are three 4-element Boolean subalgebras of any 8-element Boolean algebra, and there are three projections from $\mathbb{K}$ onto $\mathbb{Z}_2$. For each 8-element Boolean algebra assign these three projections as the maps $\ul{G}(i_{B',B})$ for the three 4-element Boolean subalgebras $B'$ of $B$ in some fashion. Then for a global section $\sigma$ of $\ul{G}$ define for an atom $P$ of $\mc{P}(\mc{H})$ the value $\sigma(B_P)$ where $B_P=\{0,P,1-P,1\}$. If $P,Q,R$ are pairwise orthogonal atoms with $B$ the Boolean subalgebra containing them, one sees using $\sigma(B)$ that either none, or exactly two, are assigned value 1. Conversely, any $\mathbb{Z}_2$-valued state can be used to build a global section. 

In \cite[p.~91]{Main} it is suggested that homological methods might be tied to the connection between the Kochen-Specker theorem and global sections of the spectral presheaf. If any progress in this direction occurs, it may also lead to progress on the $\mathbb{Z}_2$-valued state problem which remains unsolved. See also~\cite{Roumen,Caru}.
\end{remark}

We return to our program of studying results of Section~2 in the setting of restrictions to $\BSub(\mc{N})^*$. Outer daseinisation $\delta^o(\hat{A})$ of a self adjoint operator $\hat{A}$ of $\mathcal{N}$ is given in (\ref{outerdaseinisation}). For each such $\hat{A}$ this provides \cite[Thm.~7.1]{Main} a global element of the ``outer de Groot presheaf'' $\ul{\mathbb{O}}$ \cite[Def.~7.3]{Main} where $\ul{\mathbb{O}}_V=V_{sa}$ the self-adjoint operators of $V$ and $\ul{\mathbb{O}}(i_{V'V}):V_{sa}\to V'_{sa}$ is given by outer daseinisation. The restriction of $\delta^o$ to $\mathcal{V}(\mathcal{N})^*$ is realized over $\BSub(\mathcal{N})^*$ via the double commutant, with $(\delta^o)^*(\hat{A})_B=\delta^o(\hat{A})_{B''}$. The restriction of $\ul{\mathbb{O}}$ to $\mathcal{V}(\mathcal{N})^*$ is realized as a presheaf $\ul{\mathbb{O}^*}$ over $\BSub(\mc{N})^*$ where $\ul{\mathbb{O}^*}_B = (B'')_{sa}$ and $\ul{\mathbb{O}^*}(i_{B'B})(\hat{A})=\delta^o_{(B')''}(\hat{A})$. We then have the following analog of \cite[Thm.~7.1]{Main}. 

\begin{theorem}
Outer daseinisation $(\delta^o)^*$ gives an injective mapping $(\delta^o)^*:\mathcal{N}_{sa}\to\Gamma\ul{\mathbb{O}^*}$ from the self-adjoint operators of $\mathcal{N}$ to the global sections of the outer presheaf on $\BSub(\mathcal{N})^*$. 
\end{theorem}

\begin{proof}
That $(\delta^o)^*$ gives such a mapping follows from \cite[Thm.~7.1]{Main} and the fact that the ingredients involved are realizations over $\BSub(\mathcal{N})^*$ of restrictions of the outer presheaf and outer daseinisation over $\mathcal{V}(\mathcal{N})$. However, injectivity does not follow immediately, and further can not be obtained from the argument for injectivity given in \cite[Thm.~7.1]{Main} which involves infinite dimensional von Neumann subalgebras. We argue this directly. 

Suppose $\hat{A}$ and $\hat{A}'$ are distinct self-adjoint operators of $\mathcal{N}$ with $E_\lambda$ $(\lambda\in\mathbb{R})$ the spectral resolution \cite[p.~94]{Main} of $\hat{A}$ and $E_\lambda'$ $(\lambda\in\mathbb{R})$ that of $\hat{A}'$. Since $\hat{A}\neq \hat{A}'$, 
their spectral resolutions differ \cite[Thm.~5.2.4]{Kadisson}, so there is $\lambda_0\in\mathbb{R}$ with $E_{\lambda_0}\neq E_{\lambda_0}'$. Assume $E_{\lambda_0}'\nleq E_{\lambda_0}$. From the definition of a spectral resolution, $E_{\lambda_0} = \bigwedge\{E_\mu:\lambda_0<\mu\}$, so there is $\lambda_0 < \mu_0$ with $E_{\lambda_0}'\nleq E_{\mu_0}$. Let $B=\{0,E_{\mu_0}, 1-E_{\mu_0},1\}$ be the Boolean algebra of projections generated by $E_{\mu_0}$ and let $V=B''$ be its double commutant. Note that $B$ is a 4-element Boolean algebra, so belongs to $\BSub(\mathcal{N})^*$. We claim that the global sections $(\delta^o)^*(\hat{A})$ and $(\delta^o)^*(\hat{A}')$ differ at their $B$ components. 

Recall that $\delta^0(E)_V$ is the least projection in $V$ above $E$. Then since $\lambda_0<\mu_0$ we have $\bigwedge\{\delta^0(E_\mu)_V:\lambda_0<\mu\}\leq \delta^0(E_{\mu_0})_V=E_{\mu_0}$. Since $E_{\lambda_0}'\nleq E_{\mu_0}$ and spectral resolutions are order preserving, we have $E_\mu'\nleq E_{\mu_0}$ for each $\lambda_0<\mu$. Since $B$ has 4 elements, this implies that $\delta^0(E_\mu')_V\geq 1-E_{\mu_0}$ for each $\lambda_0<\mu$, and therefore $\bigwedge\{\delta^0(E_\mu')_V:\lambda_0<\mu\}\geq 1-E_{\mu_0}$. By \cite[p.~95]{Main} there are resolutions of the identity in $V$ given by 
\[\lambda\mapsto\bigwedge\{\delta^o(E_\mu)_V:\lambda<\mu\}\quad\mbox{ and } \quad\lambda\mapsto\bigwedge\{\delta^o(E_\mu')_V:\lambda<\mu\}\]
By \cite[Defn.~7.2]{Main}  and \cite[(7.21)]{Main} the first of these is the resolution of the identity of the self-adjoint operator $\delta^o(\hat{A})_V$ and the second is that of $\delta^0(\hat{A}')_V$. Since these resolutions of the identity differ at $\lambda_0$, by \cite[Thm.~5.2.4]{Kadisson} we have $\delta^o(\hat{A})_V\neq \delta^0(\hat{A}')_V$. Thus $(\delta^o)^*(\hat{A})_B\neq (\delta^0)^*(\hat{A}')_B$, showing that the global sections $(\delta^o)^*(\hat{A})$ and $(\delta^o)^*(\hat{A}')$ differ.
\end{proof}

\begin{remark}
While the mapping $\delta^0$ of the self-adjoint operators $\hat{A}$ of $\mathcal{N}$ into the global sections of the outer presheaf over $\mathcal{V}(\mathcal{N})$ is injective, by \cite[p.~99]{Main} it is not surjective. There are many other instances in topos quantum mechanics over $\mathcal{V}(\mathcal{N})$ providing an injective, but not surjective, mapping of some standard quantum mechanical notion into a the the topos formalism. 

Examples of injective, but not surjective, maps include the following. Inner daseinisation gives a map $\delta^i$ from the self-adjoint operators $\hat{A}$ of $\mathcal{N}$ into the global sections of the inner presheaf $\ul{\mathbb{I}}$ \cite[Thm.~7.1]{Main}; the map of the projection operators $\mathcal{P}(\mathcal{N})$ into the global sections of the presheaf $\mathcal{P}_{cl}\ul{\Sigma}$ of clopen subsets of the spectral presheaf \cite[Thm.~5.4]{Main}; and the mappings of pure states $|\psi\rangle$ of $\mathcal{H}$ to truth objects \cite[p.~81]{Main} and to pseudo-states \cite[p.83]{Main}. In each case, showing these maps are injective is accomplished by finding a 4-element Boolean subalgebra $B$ of projections so that they differ at the stage $V=B''$. Thus, the corresponding maps in the setting of $\BSub(\mathcal{N})^*$ are also injective. 
\end{remark}

\begin{remark}\label{clopen=power}
Restricting the presheaf $P_{cl}\ul{\Sigma}$ of clopen subsets of the spectral presheaf to the setting of $\BSub(\mathcal{N})^*$ provides considerable simplification. For a general abelian von Neumann subalgebra $V$, its spectrum $\ul{\Sigma}_V$ is a compact Hausdorff space. When $V = B''$ for a Boolean algebra of projections of $\mathcal{N}$ with at most 8 elements, the spectrum $\ul{\Sigma}_V$ is a discrete space with at most 3 elements, thus clopen subsets of $\ul{\Sigma}_V$ are simply subsets of $V$. So the restriction of $P_{cl}\ul{\Sigma}$ to $\BSub(\mathcal{N})^*$ is simply the power object of the restriction of $\ul{\Sigma}$. Thus, there is no need for a proof that this power object is a Heyting algebra as power objects are Heyting algebras in any topos. 
\end{remark}

We consider Theorem~\ref{SpecThm} in view of restrictions to $\BSub(\mathcal{N})^*$. Our argument uses only properties of restriction, and not the details specific to this situation, so applies to many such situations. Suppose $\hat{A}$ is a self-adjoint operator of $\mathcal{V}$. Theorem~\ref{SpecThm} shows that $\check{\delta}(\hat{A})$ is a natural transformation from the spectral presheaf $\ul{\Sigma}$ to the value presheaf $\ul{\mathbb{R}^\leftrightarrow}$. This means that for each $V\in\mathcal{V}(\mathcal{N})$ there is a function $\check{\delta}(\hat{A})_V:\ul{\Sigma}_V\to\ul{\mathbb{R}^\leftrightarrow}_V$, and that this collection of functions commutes with the maps $\ul{\Sigma}(i_{V'V})$ and $\ul{\mathbb{R}^\leftrightarrow}(i_{V'V})$. Restricting to subalgebras $V\in\mathcal{V}(\mathcal{N})^*$ we retain these maps and their commutativity, and moving this to the setting of the Boolean algebras $B$ of projections of such $V$ transfers these to the setting of $\BSub(\mathcal{N})^*$. Thus, we have the following essentially for free. 

\begin{theorem}\label{SpecThm*}
For a self adjoint element $\hat{A}$ of $\mathcal{V}$, we have that $\check{\delta}(\hat{A})^*$ is a natural transformation from the spectral presheaf $\ul{\Sigma}^*$ to the value presheaf $\ul{\mathbb{R}^\leftrightarrow}^*$. 
\end{theorem}

\begin{remark}
Theorems in the $\mathcal{V}(\mathcal{N})$ setting of an ``equational'' nature, such as Theorem~\ref{SpecThm}, move for free to the setting of $\BSub(\mathcal{N})^*$. The results discussed earlier in this section involving the injectivity of certain correspondences are not ``equational'' in nature, but as we have seen, are still retained upon restriction. In fact, in each case injectivity is realized at stage $V=B''$ for a 4-element Boolean algebra of projections. Another issue that is not ``equational'' is the connection between the Kochen-Specker theorem and the non-existence of a global section of the spectral presheaf. We have seen in Theorem~\ref{ShortKochen-Specker} this is retained upon restriction to $\BSub(\mathcal{H})^*$, but it would not be retained if restriction was to the at most 4-element Boolean subalgebras $\BSub(\mathcal{H})^{**}$ of projections. 
\end{remark}

We conclude this section with a view of Theorem~\ref{measures} in the context of restriction to $\BSub(\mathcal{N})^*$. This is of particular interest as this theorem specifies a bijective correspondence between a notion from the setting of $\mathcal{N}$ and one in the setting of $\mathcal{V}(\mathcal{N})$. Thus it is far from equational, and appears more substantive than the simple results involving injectivity. Further, it involves the clopen subobjects $P_{cl}\ul{\Sigma}$ of the spectral presheaf, an object that changes radically upon restriction to the $\BSub(\mathcal{N})^*$ setting to the full power object of the restriction $\ul{\Sigma}^*$ of the spectral presheaf. Still, this result is preserved by restriction.

\begin{theorem} \label{measures*}
For any von Neumann algebra $\mc{N}$ with no direct summand of type $I_2$ there is a bijection between the set of states of $\mc{N}$ and the set of measures on the presheaf $P\ul{\Sigma}^*$ of subobjects of the spectral presheaf $\ul{\Sigma}^*$ of $\BSub(\mc{N})^*$. 
\end{theorem}

\begin{proof}
By Definition~\ref{measures}, a measure $\mu$ on the presheaf $P_{cl}\ul{\Sigma}$ is a mapping $\mu:P_{cl}\ul{\Sigma}\to\Gamma\ul{[0,1]^\succeq}$ with $\mu(\ul{\Sigma})=1$ and $\mu(\ul{S_1}\vee\ul{S_2})+\mu(\ul{S_1}\wedge\ul{S_2})=\mu(\ul{S_1})+\mu(\ul{S_2})$. For a state $\rho$ of $\mc{N}$, define \cite[(14)]{Measures}

\[\mu_\rho:P_{cl}\ul{\Sigma}\to\Gamma\ul{[0,1]^\succeq}\qquad\mbox{ sending }\qquad\ul{S}\mapsto ( \rho(\hat{P}_{\ul{S}_V}))_{V\in\mc{V}(\mc{N})}\]
\vspace{-1ex}

\noindent Here $\ul{S}$ is a clopen subset of $\ul{\Sigma}$. So for each $V\in\mathcal{V}(\mathcal{N})$ we have $\ul{S}_V$ is a clopen subset of the spectrum of $V$, and therefore corresponds to a projection denoted $\hat{P}_{\ul{S}_V}$ of $V$. Then $\rho(\ul{S}_V)$ is a real number in $[0,1]$, and for $V'\subseteq V$ we have $\rho(\ul{S}_V)\leq\rho(\ul{S}_{V'})$. Thus $(\rho(\hat{P}_{\ul{S}_V}))_{V\in\mc{V}(\mc{N})}$ is an order-inverting function from $\mathcal{V}(\mc{N})$ to the real interval, so \cite[p.~6]{Measures} can be regarded as a global section of $\ul{[0,1]^\succeq}$. Restricting to the setting of $\BSub(\mc{N})^*$ preserves all this. So for each state $\rho$ of $\mc{N}$ we have a mapping $\mu_\rho^*$ from the clopen subsets $P_{cl}\ul{\Sigma}^*$ of the spectral presheaf $\ul{\Sigma}^*$ over $\BSub(\mc{N})^*$ (which by Remark~\ref{clopen=power} is the power object of $\ul{\Sigma}^*$) taking $\ul{S}$ to 
$(\rho(\hat{P}_{\ul{S}_{B''}}))_{B\in\BSub(\mc{N})^*}$.

To see that this map is onto, we follow \cite[p.~8]{Measures}. Suppose $\mu^*$ is a measure on the subobjects of $\ul{\Sigma}^*$ (all of which are clopen). Let $\hat{P}$ be any projection of $\mc{N}$ and choose a subobject $\ul{S}$ of $\ul{\Sigma}^*$ so that there is $B\in\BSub(\mc{N})^*$ containing $\hat{P}$ with $\ul{S}_B$ being the subset of the spectrum of $B$ corresponding to $\hat{P}$. This $\ul{S}$ could be chosen to be the subobject with $\ul{S}_B$ the subset corresponding to the outer daseinisation $\hat{P}^o_B$ for each $B\in\BSub(\mc{N})^*$ \cite[(10)]{Measures} since this subobject has the desired property at the 4-element subalgebra $B_{\hat{P}}=\{0,\hat{P},1-\hat{P},1\}$ that belongs to $\BSub(\mc{N})^*$. Then, as in \cite[(32)]{Measures} set 

\[ m^*(\hat{P}) = \mu^*(\ul{S})(B)=\mu^*(\ul{S}_B)\]
\vspace{-1ex}

To show $m^*$ is well defined, it must be shown that if $\ul{S},\ul{\tilde{S}}$ are two subobjects and $B,\tilde{B}$ two elements of $\BSub(\mc{N})^*$ with $\ul{S}_B$ and $\ul{\tilde{S}}_{\tilde{B}}$ corresponding to $\hat{P}$, then $\mu^*(\ul{S})(B)=\mu^*(\ul{\tilde{S}})(\tilde{B})$. In the $\mathcal{V}(\mathcal{N})$ setting, this argument is given in \cite[p.~8-9]{Measures}. This argument applies in the current situation as well. 

With $m^*(\hat{P})$ defined for each projection $\hat{P}$, we next show that it defines a finitely additive measure $m^*:\mathcal{P}(\mathcal{N})\to[0,1]$ on the projection lattice of $\mathcal{N}$. Again, the argument duplicates that of \cite[p.~9]{Measures} with the key point to show finite additivity. If $\hat{P}$ and $\hat{Q}$ are orthogonal projections, then following \cite[p.~9]{Measures} there is a context $V$ containing $\hat{P}$ and $\hat{Q}$. But $\hat{P}$ and $\hat{Q}$ being orthogonal implies that they generate an at most 8-element Boolean subalgebra $B$ of projections. So $B\in\BSub(\mathcal{N})^*$. The argument of \cite[p.~9]{Measures} then provides that $m^*(\hat{P}\vee\hat{Q})=m^*(\hat{P})+m^*(\hat{Q})$. 

Then, following \cite[p.~9]{Measures}, the measure $m^*$ on the projections of $\mathcal{N}$ defined from the measure $\mu^*$ on $P\ul{\Sigma}$ can be extended to a state $\rho_{\mu^*}$ on $\mathcal{N}$ provided there are no type $I_2$ summands. The mappings $\rho\mapsto\mu_\rho^*$ and $\mu_*\mapsto\rho_{\mu^*}$ are inverse to one another as in \cite{Measures}.
\end{proof}

\begin{remark}
We have followed the proof in \cite{Measures} closely to illustrate that not only does the result remain valid in the restricted setting, but the proof also remains valid in the restricted setting. One needs only 4-element Boolean subalgebras $B$ to find a subobject given by a daseinised projection whose value at the context $B$ is the subset of the state space of $B$ corresponding to this projection, and then $8$-element Boolean subalgebras suffice to give finite additivity. 
\end{remark}

A further result \cite{Measures} gives a point of departure between the $\mathcal{V}(\mathcal{N})$ setting and the restricted setting of $\BSub(\mathcal{N})^*$. Here, a family of clopen subobjects $\ul{S}_n$ is called {\em locally disjoint} if there is a context $V$ where the subsets $(\ul{S}_n)_V$ of the spectrum of $V$ are disjoint. A measure $\mu$ on the clopen subsets is called {\em locally $\sigma$-additive} if for each locally disjoint family of clopen subobjects $\ul{S}_n$ we have $\mu(\bigvee \ul{S}_n)(V)=\sum\mu(\ul{S}_n)(V)$. The following is given in \cite[Cor.~IV.2]{Measures}.

\begin{theorem}
The normal states of a von Neumann algebra $\mc{N}$ without type $I_2$ summand correspond to the locally $\sigma$-additive measures on the clopen subobjects. 
\end{theorem}

This result is not preserved upon restriction to the $\BSub(\mc{N})^*$ setting because the spectrum of a finite Boolean algebra is finite, so the local $\sigma$-additivity condition reduces the usual finite additivity of a measure of Definition~\ref{defn:measure}. However, the nature of this local $\sigma$-additivity condition is different than other conditions such as those in Definition~\ref{defn:measure}, and throughout the rest of the topos program, since it refers to the existence of a context $V$ at which some property holds, not to a property held at all contexts. If desired, it would likely be possible to capture normal measures in the context of $\BSub(\mc{N})^*$ via a different type of unorthodox condition based on the fact \cite[Prop.~2.10]{BertJohn} that joins of projections can be captured from the poset structure $\BSub(\mc{N})^*$. It is unclear how either of these approaches to normal measures relates to the program of formulating quantum theory within a topos. 

\section{Automorphisms}\label{sec:automorphisms}

There is a considerable amount of literature on the topics of automorphisms and group representations in the topos approach. We will not review it all, but merely discuss some basics, referring to the literature for background \cite{Doring2}. 

\begin{definition} \label{lpp} \cite[p19]{Doring2}
  For $\mc{N}$ a von Neumann algebra, an automorphism of the spectral presheaf $\ul{\Sigma}$ is a pair $(\Gamma,\tau)$ of is an essential geometric automorphism $\Gamma$ of the topos of presheaves on $\mc{V}(\mc{N})$ with inverse image functor $\Gamma^*$, and a natural isomorphism $\tau\colon\Gamma^*(\ul{\Sigma})\to\ul{\Sigma}$ whose every component $\tau_V\colon\Gamma^*(\ul{\Sigma})_V\to\ul{\Sigma}_V$ is a homeomorphism. Hence an automorphism acts as 
  \[
    \ul{\Sigma}\stackrel{\Gamma^*}{\longrightarrow} \Gamma^*(\ul{\Sigma})\stackrel{\tau}{\longrightarrow}\ul{\Sigma}\text.
  \]
  Write the group of automorphisms of $\ul{\Sigma}$ as $\Aut(\ul{\Sigma})$. 
\end{definition}

Because a poset category is Cauchy complete, each essential geometric morphism $\Gamma$ arises from an order-automorphism $\gamma$ of $\mc{V}(\mc{N})$, with the action of the inverse image functor $\Gamma^*$ on the spectral presheaf given by $\Gamma^*(\ul{\Sigma})_V=\ul{\Sigma}_{\gamma(V)}$~\cite[p7]{Doring2}. Furthermore, since $\ul{\Sigma}_V$ is the spectrum of $V$, and this determines $V$ on the nose (and not just up to isomorphism), there is a bijection between such geometric morphisms $\Gamma$ and order isomorphisms $\gamma$ of $\mc{V}(\mc{N})$. Suppose $\Gamma$ is induced by $\gamma$. For each $V\in\mc{V}(\mc{N})$, the restriction $\Gamma|V\colon V\to\gamma(V)$ is a von Neumann isomorphism, so its Gelfand transform $\mc{G}(\Gamma|V)\colon\ul{\Sigma}_{\gamma(V)}\to\ul{\Sigma}_V$ is a homeomorphism between their spectra. This gives a natural isomorphism we call $\tau_\gamma\colon\Gamma^*(\ul{\Sigma})\to\ul{\Sigma}$, and the pair $(\Gamma,\tau_\gamma)$ is an automorphism of $\ul{\Sigma}$~\cite[p8]{Doring2}. Most importantly, the following result (with different proof) can be extracted~\cite[Corollary~5.15, Proposition~5.19]{Doring2}. 

\begin{theorem} \label{lppp}
  Let $\mc{N}$ be a von Neumann algebra without type $I_2$ summand, and let $\gamma$ be an order-automorphism of $\mc{V}(\mc{N})$ with associated essential geometric morphism $\Gamma$. The natural isomorphism $\tau_\gamma$ is the unique one for which $(\Gamma,\tau_\gamma)$ is an automorphism of $\ul{\Sigma}$.
\end{theorem}
\begin{proof}
  Since $\Aut(\ul{\Sigma})$ is a group \cite[Lemma~5.2]{Doring2} under $(\Gamma_1,\tau_1)\cdot(\Gamma_2,\tau_2)=(\Gamma_2\circ\Gamma_1,\tau_1\circ\tau_2)$, it suffices to show that for $\Gamma$ the identity essential geometric morphism there is exactly one natural isomorphism $\tau\colon\ul{\Sigma}\to\ul{\Sigma}$ each component of which is a homeomorphism. Clearly, the identical natural isomorphism is one such isomorphism, we must show it is the only one. 

  We use the duality between spectra of abelian von Neumann algebras and their Boolean algebras of projections. Consider the co-presheaf $I$ on $\BSub(\mc{V})$ defined for $B$ and $B'\subseteq B$ by 
  \[ 
    I(B)=B\qquad\mbox{and}\qquad I(i_{B',B})\mbox{ is the identical embedding of $B'$ into $B$}\text.
  \]
  It suffices to show that if $\nu\colon I\to I$ is a natural isomorphism with each component $\nu_B$ an automorphism of $B$, then $\nu$ is the identity. For this, let $p$ be a projection of $\mc{N}$, generating the Boolean algebra $B_p=\{0,p,1-p,1\}$. Since $\nu_{B_p}$ is an automorphism of $B_p$ it maps $p$ to either itself or $1-P$. Since $\mc{N}$ has no type $I_2$-factor, there is an 8-element Boolean algebra $C$ containing $B_p$. In $C$ one of $p$ and $1-p$ is an atom, and the other a coatom. So the automorphism $\nu_C$ cannot map $p$ to $1-p$, and since $\nu_C$ extends $\nu_{B_p}$ also $\nu_{B_p}$ must map $p$ to itself, and so equal the identity. By  naturality of $\nu$ thus $\nu_B(p)=p$ for any Boolean algebra $B$ that contains $p$. It follows that $\nu_B$ is the identity for each $B$. 
\end{proof}

\begin{corollary}
  If $\mc{N}$ has no type $I_2$ summands, then the group $\Aut(\ul{\Sigma})$ is isomorphic to the opposite of the group $\Aut(\mc{V}(\mc{N}))$ of order-automorphisms of $\mc{V}(\mc{N})$. 
\end{corollary}
\begin{proof}
  The above discussion and theorem shows that each order-automorphism $\gamma$ of $\mc{V}(\mc{N})$ gives rise to a unique element $(\Gamma,\tau_\gamma)$ of $\Aut(\ul{\Sigma})$. That this bijective correspondence is contravariantly compatible with the group operations is not difficult to see. Note that the contravariance is a consequence of the arbitrary choice of multiplication of $\Aut(\ul{\Sigma})$ and would be covariance if the opposite choice had been made in \cite{Doring}. 
\end{proof}

Now shift attention from the poset $\mc{V}(\mc{N})$ to the poset $\mc{V}(\mc{N})^*$ of elements of height at most two in $\mc{V}(\mc{N})$. Definition~\ref{lpp} modifies to this setting obviously, giving rise to the automorphism group $\Aut(\ul{\Sigma}^*)$ of the spectral presheaf over $\mc{V}(\mc{N})^*$. 

\begin{corollary} \label{qwe}
  If $\mc{N}$ has no type $I_2$ summands, then the group $\Aut(\ul{\Sigma}^*)$ is isomorphic to the opposite of the group $\Aut(\mc{V}(\mc{N})^*)$ of order-automorphisms of $\mc{V}(\mc{N})^*$. 
\end{corollary}
\begin{proof}
  The discussion preceding Theorem~\ref{lppp} and the proof of Theorem~\ref{lppp} remain intact.
\end{proof}

To complete this line of thought, we compare the automorphism groups of the two posets. 

\begin{theorem} \label{qwf}
  For $\mc{N}$ a von Neumann algebra, restriction gives an isomorphism between the automorphism groups $\Aut(\mc{V}(\mc{N}))$ and $\Aut(\mc{V}(\mc{N})^*)$. 
\end{theorem}
\begin{proof}
  We may work instead with the posets $\BSub(\mc{N})$ and $\BSub(\mc{N})^*$, which are are isomorphic to $\mc{V}(\mc{N})$ and $\mc{V}(\mc{N})^*$.
  If $\gamma$ is an automorphism of $\BSub(\mc{N})$, it must preserve the height of elements, so the restriction $\gamma|\BSub(\mc{N})^*$ is an automorphism of $\BSub(\mc{N})^*$. Restriction preserves composition. Suppose $\gamma$ and $\gamma'$ are automorphisms of $\BSub(\mc{N})$ whose restrictions to $\BSub(\mc{N})^*$ agree. 
  Then $\gamma$ and $\gamma'$ are equal, because each element of $\BSub(\mc{N})$ is the join of the atoms beneath it. It remains to show that each automorphism $\phi$ of $\BSub(\mc{N})^*$ arises as the restriction of an automorphism of $\BSub(\mc{N})$. 

  We make use of results of \cite{HHLN}. We will use $P$ for the poset of elements of $\BSub(\mc{N})$ of height at most three (this was called $\BSub(\mc{N})^*$ in \cite{HHLN}). Atoms of $P$ are called points, the elements of height two lines, the elements of height three planes, and the poset $P$ is the hypergraph $\mc{G}(\mc{P}(\mc{N}))$ associated to the orthomodular lattice $\mc{P}(\mc{N})$. Since $\mc{P}(\mc{N})$ is an orthomodular lattice, and hence an orthomodular poset, $P$ can be constructed from its elements $\BSub(\mc{N})^*$ of height at most two as follows~\cite[Thm.~6.8]{HHLN}. For each downset of $\BSub(\mc{N})$ that is isomorphic to the elements of height at most two in the lattice of subalgebras of a 16-element Boolean algebra, insert an element of height three above all the elements in this downset. In the language of \cite{HHLN}, every configuration that looks like a plane is a plane. It follows that the automorphism $\phi$ of $\BSub(\mc{N})^*$ extends to an automorphism $\hat{\phi}$ of $P$. 

  Recall the notions of a tall and short orthodomain~\cite[Definition~5.13]{HHLN}. Since each element of $P$ has height at most three, $P$ is a short orthodomain, and it follows from~\cite[Proposition~5.14]{HHLN} that $\BSub(\mc{N})$ is a tall orthodomain. For each $x$ in $\BSub(\mc{N})$, it follows from \cite[Definitions~5.1 and~5.9]{HHLN} that the set ${\downarrow}x\cap P$ is a Boolean shadow. By definition of a tall orthodomain, there is a bijective correspondence between elements $x$ of $\BSub(\mc{N})$ and Boolean shadows $S$ of $P$ where the element $x$ corresponds to the Boolean shadow ${\downarrow}x\cap P$ and the Boolean shadow $S\subseteq P$ corresponds to its join $\bigvee S$ in $\BSub(\mc{N})$. Since automorphisms of $P$ take Boolean shadows to Boolean shadows, $\hat{\phi}$ can be extended to an automorphism $\gamma$ of $\BSub(\mc{N})$ by setting $\gamma(\bigvee S)=\bigvee \gamma (S)$ for each Boolean shadow $S$. Clearly the restriction of $\gamma$ to $\BSub(\mc{N})^*$ is $\phi$. 
\end{proof}

Write $\Aut(\mc{P}(\mc{N}))$ for the automorphism group of the orthomodular lattice $\mc{P}(\mc{N})$, and $\Aut_{\mathrm{Jordan}}(\mc{N})$ for the group of automorphisms of $\mc{N}$ with its Jordan product, and $\Aut(\mc{N}_{\mathrm{part}})$ for the group of automorphisms of $\mc{N}$ with its partial structure~\cite[Definitions~4.3 and~4.5]{Doring2}. 

\begin{corollary}
  If $\mc{N}$ is not isomorphic to $\mathbb{C}^2$ and has no type $I_2$-summand, then all of the groups $\Aut(\mc{V}(\mc{N}))$, $\Aut(\mc{V}(\mc{N})^*)$, $\Aut(\mc{P}(\mc{N}))$, $\Aut_{\mathrm{Jordan}}(\mc{N})$ and $\Aut(\mc{N}_{\mathrm{part}})$ are isomorphic, and they are contravariantly isomorphic to the groups $\Aut(\ul{\Sigma})$ and $\Aut(\ul{\Sigma}^*)$. Furthermore, there is an embedding of the automorphism group $\Aut(\mc{N})$ of $\mc{N}$ into each of the first four groups, and a contravariant embedding of it into the latter two. 
\end{corollary}
\begin{proof}
  Combine Corollary~\ref{qwe} and Theorem~\ref{qwf} with~\cite[Corollary~5.15]{Doring2}. 
\end{proof}

In the process of considering further the use of automorphisms in the treatment of \emph{flows} in~\cite{DoringFlows}, we had trouble with the embedding of the automorphism group of the spectral presheaf into the automorphism group of the Heyting algebra of clopen subobjects of the spectral presheaf \cite[Proposition~4.3]{Doring}. The indicated formula, described also in equation (38) does give a homomorphism between the automorphism groups, but injectivity is not shown, and it seems that it maps each automorphism to the identity automorphism of the Heyting algebra of clopen subobjects. Of course, this would be the same also over the poset $\mc{V}(\mc{N})^*$, which is our point.


\section{Concluding remarks}\label{sec:conclusion}

We have shown that many of the core results in the topos approach over $\mc{V}(\mc{N})$ are retained upon restriction to a topos over the simpler poset $\mc{V}(\mc{N})^*$. Further, aspects of the topos approach simplify considerably in the $\mc{V}(\mc{N})^*$ setting. It is not clear what is gained in the more general setting that is not gained in the simpler one, or if the topos over the simpler object is sufficient. 

The toposes in the two settings are different. That the difference is not felt by restriction comes from the following fact. Many results of the topos approach involve embedding of physical notions such as projections into a certain collection of presheaves, and generally these embeddings are not onto. The role of the presheafs that do not come from their actual counterparts is not clear, and this applies in both settings. The two settings will have different supplies of these unfelt presheaves (perhaps the $\mc{V}(\mc{N})$ setting has more, we do not know), but the importance of this is not clear. 

The logic of the toposes in the two settings differs too. Recall that the Heyting algebra of downsets of a poset $P$ has height:
\begin{itemize}
  \item 0 precisely when it satisfies the formula $\varphi_0$ given by $((y\to x)\to y)\to y$;
  \item at most 1 precisely when it satisfies the formula $\varphi_1$ given by $((z\to\varphi_0)\to z)\to z$;
  \item at most 2 precisely when it satisfies the formula $\varphi_2$ given by $((w\to\varphi_1)\to w)\to w$.
\end{itemize} 
The poset $\mc{V}(\mc{N})$ is of infinite height if and only if $\mc{N}$ is infinite-dimensional, so its downsets will not satisfy any of these conditions. But $\mc{V}(\mc{N})^*$ always has height at most 2, so satisfies the third condition. This is reflected in the equational properties of the logic of the topos. The connection between these laws of intuitionistic logic that separate the toposes over $\mc{V}(\mc{N})$ and over $\mc{V}(\mc{N})^*$ and properties of the quantum mechanical system represented by $\mc{N}$ is not apparent. 

Finally, the topos approach of Heunen, Landsman and Spitters~\cite{HLS} more generally considers a C*-algebra $\mc{N}$ instead of a von Neumann algebra, and the poset $\mc{C}(\mc{N})$ of all commutative C*-subalgebras instead of $\mc{V}(\mc{N})$. In general $\mc{C}(\mc{N})$ and $\mc{V}(\mc{N})$ are very different, because a von Neumann algebra $\mc{N}$ can have more abelian C*-subalgebras than abelian von Neumann subalgebras. But $\mc{C}(\mc{N})^*$ equals $\mc{V}(\mc{N})^*$ because all elements of $\mc{C}(\mc{N})^*$ are finite-dimensional and hence von Neumann algebras themselves. Therefore one might consider the topos based on $\mc{C}(\mc{N})^*$ for arbitrary C*-algebras $\mc{N}$. However, it is hard to compare this to the topos based on $\mc{C}(\mc{N})$, because unlike for von Neumann algebras $\mc{N}$, a C*-algebras $\mc{N}$ may not have many projections, and all techniques based on Boolean algebras that we used cannot be applied.

In conclusion, consideration of the topos over $\mc{V}(\mc{N})^*$ rather than $\mc{V}(\mc{N})$ raises a number of specific questions about the topos approach that may lead to a better understanding of it. If it turns out that the topos over $\mc{V}(\mc{N})^*$ is sufficient, this may allow considerable simplification of aspects of the topos approach, and also open the door to the geometrical techniques to treat the posets $\mc{V}(\mc{N})^*$ and the morphisms between them as described in \cite{HHLN,BertJohn}.


\newpage

\begin{thebibliography}{99}

\bibitem{Caru} G. Car{\`u}, {\em On the cohomology of contextuality}, Proceedings of QPL, EPTCS 236:21--39, 2017.

\bibitem{Measures} A D\"oring, {\em Quantum states and measures on the spectral presheaf}, arXiv:0809.4847, 2008.  

\bibitem{Doring3} A. D\"oring, {\em Spectral presheaves as quantum state spaces}, Philos. Trans. Roy. Soc. A 373 no. 2047, 2015.

\bibitem{Doring4} A. D\"oring, {\em Topos-based logic for quantum systems and bi-Heyting algebras}, Logic and algebraic structures in quantum computing, 151?173, Lect. Notes Log., 45, Assoc. Symbol. Logic, La Jolla, CA, 2016.

\bibitem{Doring2} A. D\"oring and  B. Dewitt, {\em Self-adjoint operators as functions I}, Comm. Math. Phys. 328 (2), 499--525, 2014.

\bibitem {Doring} A. D\"oring and J. Harding, {\em Abelian subalgebras and the Jordan structure of a von Neumann algebra}, Houston J. Math. 42(2):559--568, 2016.

\bibitem{Main} A. D\"oring and C. J. Isham, {\em `What is a thing?': topos theory in the foundations of physics}, arXiv:0803.0417, 2008. 

\bibitem{Main2} A. D\"oring and C. J. Isham, {\em `What is a thing?': topos theory in the foundations of physics}, New structures for physics, 753?937, Lecture Notes in Phys., 813, Springer, Heidelberg, 2011.

\bibitem{Doring1} A. D\"oring, A. and C. J. Isham, {\em Classical and quantum probabilities as truth values}, J. Math. Phys. 53 (3), 2012. 

\bibitem{Doring2}  A. D\"oring, {\em Generalized Gelfand spectra}, arXiv:1212.2613v2, 2013. 

\bibitem{DoringFlows} A. Doring, {Flows on generalized Gelfand spectra of nonabelian unital C$^*$-algebras and time evolution of quantum systems}, arXiv:1212.4882v2, 2013. 

\bibitem{Faure}  C. Faure and A. Fr\"olicher, {\em Modern projective geometry}, Mathematics and its Applications, 521. Kluwer Academic Publishers, Dordrecht, 2000.

\bibitem{Flori} C. Flori, {\em Review of the topos approach to quantum theory}, ArXiv:1106.5660, 2011.

\bibitem{Flori1} C. Flori, {\em A first course in topos quantum theory}, Lecture Notes in Physics 868, Springer-Verlag, 2013. 

\bibitem{Flori2} C. FLori, {\em A second course in topos quantum theory}, Lecture Notes in Physics 944, Springer-Verlag, 2018. 

\bibitem{HHLN} J. Harding, C. Heunen, A. J. Lindenhovius, and M. Navara, {\em Boolean subalgebras of orthoalgebras}, arXiv:1711.03748, 2018.

\bibitem{JohnDerek} J. Harding, E. Jager, and D. Smith, {\em Group-valued measures on the lattice of closed subspaces of a Hilbert space}, Internat. J. Theoret. Phys. 44:539--548, 2005. 

\bibitem{BertJohn} J. Harding and A. J. Lindenhovius, {\em Hypergraphs and AW*-algebras}, manuscript. 

\bibitem{Heunen} C. Heunen, {\em The many classical faces of quantum structures}, Entropy 19:144, 2017.

\bibitem{HLS} C. Heunen, N. P. Landsman, and B. Spitters, {\em A Topos for Algebraic Quantum Theory}, Communications in Mathematical Physics 291. 63-110. 2009.

\bibitem{HLS2} C. Heunen, N. P. Landsman, and B. Spitters, {\em Bohrification of Operator Algebras and Quantum Logic}, Synthese. http://dx.doi.org/10.1007/s11229-011-9918-4. 2011.

\bibitem{HLS3} C. Heunen, N. P. Landsman and B. Spitters, {\em Bohrification. Deep Beauty: Mathematical Innovation and Research for Underlying Intelligibility in the Quantum World}, ed. Hans Halvorson. Cambridge University Press. 2011.

\bibitem{Isham}  C. J. Isham and J. Butterfield, {\em A Topos perspective on the Kochen-Specker theorem. I. Quantum states as generalized valuations}, Internat. J. Theoret. Phys. 37:2669--2733, 1998.

\bibitem{Isham2} C. J. Isham and J. Butterfield, {\em A Topos perspective on the Kochen-Specker theorem. II. Conceptual aspects, and quantum analogs}, Internat. J. Theoret. Phys. 38:827--859, 1999.

\bibitem{Isham3} C. J. Isham, J. Hamilton, and J. Butterfield, {\em A Topos perspective on the Kochen-Specker theorem. III. Von Neumann algebras as the base category}, Internat. J. Theoret. Phys. 39:1413--1436, 2000.

\bibitem{Isham4} C. J. Isham and J. Butterfield, {\em A Topos perspective on the Kochen-Specker theorem. IV. Interval valuations}, Internat. J. Theoret. Phys. 41:613--639, 2002.

\bibitem{Kadisson} R. V. Kadison and J. R. Ringrose, Fundamentals of the Theory of Operator Algebras. Vol. I: Elementary Theory. New York, Academic Press 1983.

\bibitem{Kochen} S. Kochen and E. P. Specker, {\em Logical structures arising in quantum theory}, The theory of models, ed. J. W. Addison, L. Henkin and A. Tarski, North-Holland, Amsterdam, 1965.

\bibitem{Roumen} F. Roumen, {\em Cohomology of effect algebras}, Proceedings QPL, EPTCS 236:174-201, 2017.

\bibitem{Wolters} S. Wolters, {\em A comparison of two topos-theoretic approaches to quantum theory}, arXiv:1010.2031, 2011. 


\end{thebibliography}
\end{document}